\documentclass[format=acmsmall, review=false]{acmart}

\usepackage{acm-ec-19}
\usepackage{booktabs} 

\usepackage{amssymb}
\usepackage{verbatim}
\usepackage{bm}
\usepackage{tikz}
\usetikzlibrary{positioning,arrows}
\usepackage{caption,subcaption}
\usepackage{color}

\theoremstyle{definition}

\theoremstyle{remark}
\newtheorem*{reminder}{Reminder}

\begin{document}
\title[The Complexity of SPR matching]{The Complexity of Student-Project-Resource Matching-Allocation Problems}
\author{Anisse Ismaili}


\begin{abstract}
\vspace*{1cm}    
In this work, we consider a three sided student-project-resource matching-allocation problem, in which students have preferences on projects, and projects on students. While students are many-to-one matched to projects, indivisible resources are many-to-one allocated to projects whose capacities are thus endogenously determined by the sum of resources allocated to them. Traditionally, this problem is divided into two separate problems: (1) resources are allocated to projects based on some expectations (resource allocation problem), and (2) students are matched to projects based on the capacities determined in the previous problem (matching problem). Although both problems are well-understood, unless the expectations used in the    first problem are correct, we obtain a suboptimal outcome. Thus, it is desirable to solve this problem as a whole without dividing it in two.  

Here, we show that finding a nonwasteful matching is $\text{FP}^{\text{NP}}[\text{log}]$-hard, and deciding whether a stable matching exists is $\text{NP}^{\text{NP}}$-complete. These results involve two new problems of independent interest: \textsc{ParetoPartition}, shown $\text{FP}^{\text{NP}}[\text{poly}]$-complete and strongly $\text{FP}^{\text{NP}}[\text{log}]$-hard, and \textsc{$\forall\exists$-4-Partition}, shown strongly $\text{NP}^{\text{NP}}$-complete. 
\end{abstract}

\maketitle

\vspace*{1cm}   

\section{Model}

In this section, we introduce necessary definitions and notations.

\begin{definition}[Student-Project-Resource (SPR) Instance]
\label{def:SPRinstance}
It is a tuple $(S,P,R,\succ_S,\succ_P,T_R,q_R)$.
\begin{itemize}
  \item $S=\{s_1, \ldots, s_{|S|}\}$ is a set of students.
  \item $P=\{p_1, \ldots, p_{|P|}\}$ is a set of projects.
  \item $R=\{r_1, \ldots, r_{|R|}\}$ is a set of resources.
  \item $\succ_S = (\succ_s)_{s \in S}$ are the students'
preferences over set $P\cup\{\emptyset\}$.
  \item $\succ_P =(\succ_p)_{p \in P}$ are the projects' 
preferences over set $S\cup\{\emptyset\}$.
\item Resource $r$ has capacity $q_r\in\mathbb{N}_{>0}$, and $q_R=(q_r)_{r\in R}$.
  \item Resource $r$ is compatible with $T_r\subseteq P$, and $T_R=(T_r)_{r\in R}$.
\end{itemize}
 For soundness,\footnote{Without these properties, this work is still valid, though a claiming or envious pair $(s,p)$ may not necessarily make sense.} 
 every preference $\succ_p$ may extend to $2^S$ in a
 non-specified manner such that:
 \begin{itemize}
\setlength{\itemsep}{0.5em}
 \item $\forall s,s'\in S, \forall S'\subseteq S\setminus\{s,s'\}, 
 s \succ_p s' \Leftrightarrow S'\cup \{s\} \succ_p S'\cup \{s'\}$ (responsiveness)
 and 
 \item  $\forall s\in S,\forall S'\subseteq S\setminus\{s\}, s\succ_p\emptyset \Leftrightarrow S'\cup \{s\}\succ_pS'$ (separability).
 \end{itemize}
\end{definition}

Contract $(s,p)\in S \times P$ means that student $s$ is matched to project $p$.
Contract $(s,p)$ is acceptable for student $s$ (resp. project $p$) if $p
\succ_s \emptyset$ holds (resp. $s \succ_p \emptyset$). The contract is acceptable
when both hold. 
W.l.o.g., we define set of contracts $X\subseteq S\times P$ by $(s, p) \in X$ if and only if
it is acceptable for $p$.\footnote{%
   For designing a strategyproof mechanism, we assume each $\succ_s$ is
   private information of $s$, while the rest of parameters are public. Thus, $X$ does not need to be part of the input, since it is characterized by projects' preferences.} 

\begin{definition}[Matching]
    A matching is a subset $Y\subseteq X$, 
    where for every student $s\in S$, subset $Y_s=\{(s,p)\in Y\mid p\in P\}$ satisfies $|Y_s|\leq 1$, and either 
   \begin{itemize}
   \item $Y_s=\emptyset$, or
   \item $Y_s=\{(s,p)\}$ and $p\succ_s\emptyset$, holds.
   \end{itemize}
For a matching $Y$, let $Y(s)\in P\cup\{\emptyset\}$ denote the project $s$ is matched,
and $Y(p)\subseteq S$ denote the set of students assigned to project $p$.
\end{definition}

\begin{definition}[Allocation]
An allocation $\mu:R\rightarrow P$ maps each resource $r$ to a project $\mu(r)\in T_r$. 
(A resource is indivisible.)
Let $q_{\mu}(p)=\sum_{r \in \mu^{-1}(p)} q_r$.\footnote{For $\mu^{-1}(p)=\emptyset$, we assume that an empty sum equals zero.}
\end{definition}

\begin{definition}[Feasibility]
A feasible matching $(Y,\mu)$ is a couple of a matching and an allocation where
for every project $p\in P$, it holds that $|Y(p)| \leq q_{\mu}(p)$.
\end{definition}
In other words, matching $Y$ is feasible with allocation $\mu$ if
each project $p$ is allocated enough resources by $\mu$ to accommodate $Y(p)$.
We say $Y$ is feasible if there exists $\mu$ such
that $(Y, \mu)$ is feasible.

Traditionally (e.g. with fixed quotas), for feasible matching $(Y, \mu)$ and $(s,p)\in X\setminus Y$, we say student $s$ \emph{claims an empty seat} of $p$ if $p \succ_s Y(s)$ and matching $Y \setminus \{(s,Y(s))\} \cup \{(s, p)\}$ is feasible with \emph{same} allocation $\mu$. However, in our setting \cite{Goto:AEJ-micro:2016}, since the distributional constraint is endogenous and as flexible as allocations are, the definition of nonwastefulness uses this flexibility, as follows.
\begin{definition}[Nonwastefulness]
Given feasible matching $(Y,\mu)$, 
a contract $(s,p)\in X\setminus Y$ is a claiming pair
if and only if: 
\begin{itemize}
\item student $s$ has preference $p \succ_s Y(s)$, and 
\item matching $Y\setminus \{(s,Y(s))\}\cup\{(s,p)\}$ is feasible with some \emph{possibly new} allocation $\mu'$.
\end{itemize} 
A feasible matching $(Y,\mu)$ is nonwasteful if it has no claiming pair.
\end{definition}


  In other words, $(s, p)$ is a claiming pair 
  if it is possible to move $s$ to a more preferred project $p$
  while
  keeping the assignment of other students unchanged with allocation
  $\mu'$. Note that $\mu'$ can be different from $\mu$.
  Thus, $(s, p)$ can be a claiming pair even if moving her to $p$ is
  impossible with the current allocation $\mu$, but it becomes
  possible with a different/better allocation $\mu'$.

\begin{definition}[Fairness]\label{def:fair}
Given feasible matching $(Y,\mu)$, 
contract $(s,p)\in X\setminus Y$ is an envious pair
  if and only if: 
\begin{itemize}
\item student $s$ has preference $p \succ_s Y(s)$, and 
\item there exists student $s'\in Y(p)$ such that $p$ prefers
$s \succ_p s'$.\footnote{Note that matching $(Y\setminus \{(s,Y(s)),(s',Y(s'))\})\cup\{(s,p)\}$ is still feasible with same allocation $\mu$.}
\end{itemize}
We also say $s$ has justified envy toward $s'$ when the above
conditions hold.
A feasible matching  $(Y,\mu)$ is fair if it has no envious pair
(equivalently, no student has justified envy).
\end{definition}
In other words, student $s$ has justified envy toward $s'$, 
  if $s'$ is assigned to project $p$, although 
  $s$ prefers $p$ over her current project $Y(s)$
  and project $p$ also prefers $s$ over $s'$.

\begin{definition}[Stability]
A feasible matching $(Y,\mu)$ is stable if it is nonwasteful and fair (no claiming/envious pair). 
\end{definition}

\begin{definition}[Pareto Efficiency]
    Matching 
    $Y$ is Pareto dominated by 
    $Y'$ if all students weakly prefer
    $Y'$ over $Y$ and at least one student strictly prefers $Y'$.
    A feasible matching
    is Pareto efficient if
    no feasible matching 
    Pareto dominates it. 
\end{definition}
Pareto efficiency implies nonwastefulness (not vice versa).

\begin{definition}[Mechanism]
    Given any SPR instance, {
      a mechanism} 
    outputs a feasible
matching $(Y, \mu)$. 
If a mechanism always obtains a feasible matching that satisfies 
property A (e.g., fairness), we say this mechanism
    is A (e.g., fair).
      A mechanism is strategyproof if no student gains by
reporting a preference different from her true one.
\end{definition}

An SPR belongs to a general class of problems,
where distributional constraints satisfy a condition called 
\emph{heredity}\footnote{Heredity means
  that if matching $Y$ is feasible, then any of its subsets are also
  feasible. An SPR satisfies this property.}
\citep{Goto:AEJ-micro:2016}.
 Two general strategyproof mechanisms exist in this context \citet{Goto:AEJ-micro:2016}.
\emph{First}, Serial Dictatorship (SD) obtains a Pareto efficient
(thus also nonwasteful) matching.
{%
SD matches students one by one, based on a fixed ordering.
Let $Y$ denote the current (partial) matching.
For next student $s$ from the fixed order, SD chooses $(s,p) \in X$ and add it to
$Y$, where $p$ is her most preferred project
  s.t. $Y\cup\{(s,p)\}$ is
feasible with some allocation $\mu'$.}
  Unfortunately,
SD is computationally expensive\footnote{%
  It requires to solve \textsc{SPR/FA} (see below) $O(|X|)$ times.}
and unfair.
%
\emph{Second}, Artificial Caps Deferred Acceptance
(ACDA) obtains a fair matching
in polynomial-time. The idea is to fix a resource allocation $\mu$
 and
 run the well-known Deferred Acceptance
(DA)~\citep{Gale:AMM:1962}.
In DA, each student first applies to her most preferred project.
Then each project deferred accepts applicants up to
its capacity limit based on its preference and the rest of the students are rejected.
Then a rejected student applies to her
second choice, and so on.\footnote{%
Each project deferred accepts
applying students, without distinguishing newly applied and
already deferred accepted students.}
However, ACDA is 
inefficient since {$\mu$ is chosen independently
  from students' preferences.}

%
\begin{example}\label{counterex:stable}
Nonwastefulness and fairness are incompatible
since there exists an instance with no stable matching.
Let us show a simple example with two students $s_a, s_b$,
two projects $p_a, p_b$, and a unitary resource compatible with both.
Students' preferences are $p_a\succ_{s_a}p_b$ and $p_b\succ_{s_b}p_a$.
Projects' are $s_b\succ_{p_a}s_a$ and $s_a\succ_{p_b}s_b$.
By symmetry, assume the resource is
allocated to $p_a$. From fairness,
$s_b$ must be allocated to $p_a$. Then $(s_b, p_b)$ becomes
a claiming pair.\footnote{We use this example
    as a building block in the next section.}
\end{example}


\section{The Complexity of SPR}

In this section, we study the computational complexity of the problems defined below.

\noindent\begin{definition}[Computational problems]~
\begin{itemize}
\item \textsc{SPR/FA}:
Given an SPR instance and a matching $Y$,
does an allocation $\mu$ exist such that $(Y,\mu)$ is a feasible matching?
\item \textsc{SPR/Nw/Verif}:
Given an SPR instance and a feasible matching $(Y,\mu)$, is it nonwasteful? 
\item \textsc{SPR/Nw/Find}: 
Given an SPR instance, find a nonwasteful matching $(Y,\mu)$. 
\item \textsc{SPR/Stable/Verif}: 
Given an SPR instance and a feasible matching,
is it stable? 
\item \textsc{SPR/Stable/Exist}: 
Given an SPR instance, does a stable matching exist?
\end{itemize}
\end{definition}

\begin{reminder}[Computational Complexity]
We assume the following common knowledge: 
(decision) problem, length function, 
classes P, NP, complementation,
hardness and completeness.  
An SPR instance has length $\Theta(|S||P|+|P||R|)$.

A number problem is said \emph{strongly} hard if its hardness holds even when restricting to instances whose numbers are polynomially bounded.
For instance, NP-complete problem $\textsc{Partition}$ (as well as $\textsc{SubsetSum}$ or $\textsc{Knapsack}$) admits an algorithm polynomial in its largest number; hence, it is not strongly hard.
However, problem $\textsc{4-Partition}$ is NP-hard even when its numbers are polynomially bounded \cite{garey1979computers}. Therefore, it is a \emph{strongly} NP-hard problem.

While a decision problem only allows for one $\{0,1\}$ (no/yes) output, 
a function problem allows for an entire $\{0,1\}$-word (hence, any finite discrete object, or w.l.o.g. an integer).
A function problem in class $\text{FP}^\text{NP}[\text{poly}]$ (resp. $\text{FP}^\text{NP}[\text{log}]$) can be solved by a polynomial (resp. logarithmic) number of calls to an NP-oracle.
Typically, any optimization problem whose decision version (whether a solution better than a threshold exists) is in NP, is in $\text{FP}^\text{NP}[\text{poly}]$ or $\text{FP}^\text{NP}[\text{log}]$: one finds the optimum by a binary search that calls the decision version. It is usually polynomial in the number of bits for numbers, but when instances have no numbers then the binary search is typically logarithmic.
Hardness in these classes is induced by metric reductions from function problem $\Pi$ to $\Pi'$,
where finding the output for $\Pi'$ in polynomial time provides the output for $\Pi$ in polynomial time.

Class NP is the class of problems whose yes-instances admit a certificate (e.g. a solution) that can be verified in polynomial-time. When the verification procedure requires an NP-oracle, the problem is in class $\text{NP}^\text{NP}$. Class coNP (resp. $\text{coNP}^\text{NP}$) is the complement of class NP (resp. $\text{NP}^\text{NP}$).
\end{reminder}

Let us start by simply observing how brute-force methods depend on the parameters of these problems.
At first glance, there are $O(|P|^{|S|})$ matchings and $O(|P|^{|R|})$ resource allocations.
Whether a matching $Y$ is feasible by some allocation can be decided using dynamic programming 
on subproblems $T_k(\kappa_1,\ldots,\kappa_{|P|})\in\{\text{false},\text{true}\}$ (for integers $0\leq k\leq |R|$ and $0\leq \kappa_p\leq |S|$) which ask whether some allocation can provide $\kappa_p$ seats for each $p\in P$, using only resources $\{r_1,\ldots,r_k\}$.
There are $O\left(|R||S|^{|P|}\right)$ subproblems.
Each subproblem can be solved in time $O(|P|)$ by the following recurrence. First, 
$T_0(\bm{\kappa})=\left\{\text{true if }\bm{\kappa}\equiv 0, \text{false otherwise}\right\}$, 
and second, for $k>0$ and $\bm{\kappa}\in[0,n]^p$, $T_k(\bm{\kappa})=\bigvee\nolimits_{{i=1}\mid{\kappa_i\geq q_{r_k}}}^{p} T_{k-1}(\kappa_1,\ldots,\kappa_i-q_{r_k},\ldots,\kappa_p)$, both hold.
Therefore, the dynamic program takes time $O(|S|^{|P|}|P||R|)$, including a last iteration that queries for an allocation with at least the required numbers of seats (rather than exactly). 
Consequently:
\begin{itemize}
\item \textsc{SPR/FA} can be decided in time $O(|S|^{|P|}|P||R|)$,
\item \textsc{SPR/Nw/Find} can be solved in time $O(|S|^{|P|+1}|P|^2|R|)$ by mechanism SD, and
\end{itemize}
these two problems are XP-tractable with respect to parameter $|P|$ (while stability seems harder to decide).
A verification problem \textsc{SPR/FA} decidable in polynomial-time would contain our problems to class NP. However, this is not the case in general, as the following theorem shows.

\begin{theorem}
\label{thm:feasibility}
\textsc{SPR/FA} is NP-complete.
\end{theorem}

\begin{proof}
Since an allocation $\mu$ that makes $(Y,\mu)$ a feasible matching is an efficiently verifiable certificate for yes-instances, \textsc{SPR/FA} belongs to NP.
For hardness, any instance of \textsc{4-Partition}, defined by positive integers multiset $W=\{w_1,\ldots,w_{4m}\}$ and target $\theta\in\mathbb{N}$ (with $\sum_{w\in W}w=m\theta$ and $\forall i\in[4m],\frac{\theta}{5}<w_i<\frac{\theta}{3}$) is reduced to an instance of \textsc{SPR/FA} with $m$ projects $p_1,\ldots,p_m$. 
In the given matching $Y$, $\theta$ students are matched to each project.
Resources $R$ are identified with weights $W$:
$q_R=(w_1,\ldots,w_{4m})$ and $T_r= P$ for every $r\in R$. 
The correspondence is straightforward between a partition of $W$ into $m$ subsets of size $4$ that hit $\theta$, and an allocation with capacity for $\theta$ students on $m$ projects (hence $4$ resources per project).
Crucially, since \textsc{4-Partition} is NP-hard even if its integers are polynomially bounded, so is the number of students and the reduction is polynomial. 
\end{proof}

%

Intricate complexity results follow from the hardness of feasibility.\footnote{It tends to push problems to be strictly harder than NP.}
Also, in Th. \ref{thm:feasibility}, the \emph{strong} NP-hardness of \textsc{4-Partition} is necessary: 
a similar construction from \textsc{Partition} with two projects would require exponentially many students, hence the reduction would not be polynomial.
Therefore, we need to create \textsc{ParetoPartition} and \textsc{$\forall\exists$-4-Partition} and show them \emph{strongly} hard, so our reductions have polynomially many students.

\noindent\begin{definition}[New fundamental problems]~
\begin{itemize}
\item \textsc{ParetoPartition}:\\
Given positive integer multiset $W=\{w_1,\ldots,w_{|W|}\}$, a number $m\in\mathbb{N}$ of desired subsets, and target $\theta\in\mathbb{N}$, any partition of $W$ into a list $V_1,\ldots,V_{m}$ of $m$ subsets is mapped to deficit vector $\bm{\delta}\in\mathbb{Z}^m$ that is defined for every\footnote{$[m]$ is shorthand of $\{1, \ldots, m\}$.} $i\in[m]$ by: 
$$\delta_i=\min\left\{w(V_i)-\theta,0\right\},$$ 
where $w(V_i)=\sum_{w\in V_i}w$. 
(Subset $V_i$ has negative deficit if it sums below $\theta$, and deficit zero if it surpasses $\theta$.)
The problem is to find one partition of $W$ into $m$ subsets whose deficit vector $\bm{\delta}$ is Pareto
efficient\footnote{Given two vectors $\delta,\delta'\in\mathbb{Z}^m$,
  vector $\delta$ Pareto dominates $\delta'$ if and only if: $\forall
  i\in[m],\delta_i\geq\delta'_i$ and $\exists
  i\in[m],\delta_i>\delta'_i$.
  For a set of vectors $\Delta$ and $\delta \in \Delta$,
    $\delta$ is 
  Pareto efficient in $\Delta$ when no other vector $\delta' \in
  \Delta$
  Pareto dominates it.}
within the deficit vectors of all partitions of $W$.
\item \textsc{$\forall\exists$-4-Partition}:\\
Given target $\theta\in\mathbb{N}$, list of integers $W=(w_1,\ldots,w_{4m})$ s.t. $\frac{\theta}{5}\!<\!w_i\!<\!\frac{\theta}{3}$ and list of disjoint couples $\mathcal{L}=(u_1,v_1),\ldots,(u_\ell,v_\ell)$ from $W$, for map $\sigma:[\ell]\rightarrow\{0,1\}$, a partition of $W$ into $m$ subsets $V_1,\ldots,V_m$ is $\sigma$-satisfying if and only if:
\begin{itemize}
\setlength{\itemsep}{0.25em}
\item $\forall i\in[m]$, $|V_i|=4$ and $w(V_i)=\theta$, 
\item $\forall i\in[\ell]$, $u_i\in V_i$ \enskip and\enskip $\forall i\in[\ell]$, $v_i\in V_i$ if and only if $\sigma(i)=1$. 
\end{itemize}
(Thus, $u_i$ and $v_i$ are together in $V_i$ if and only if $\sigma(i)=1$.)
The question is: Does, for every map $\sigma:[\ell]\rightarrow\{0,1\}$, a $\sigma$-satisfying partition of $W$ into $m$ subsets exist?
\end{itemize}
\end{definition}

\subsection{The Complexity of Nonwastefulness}

Here we first show that there is no natural verification procedure that would make computing a nonwasteful matching\footnote{whose existence is guaranteed by mechanism SD} belong to NP. Indeed, we then show that \textsc{SPR/Nw/Find} is $\text{FP}^{\text{NP}}[\text{log}]$-hard: one can embed a logarithmic number of calls to SAT in a single call to \textsc{SPR/Nw/Find}, which is strictly harder than NP. 

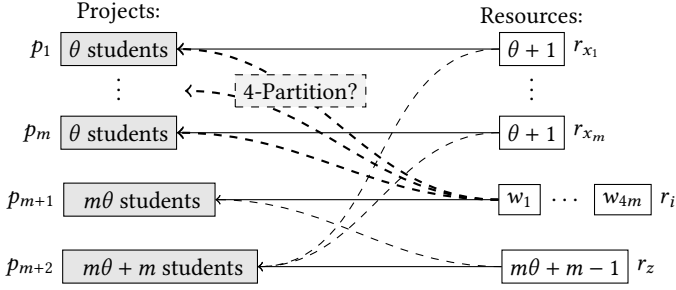
\begin{figure}[t]
\centering
\begin{tikzpicture}[scale=1.1, every node/.style={scale=0.9}]
\node at (0,-0.8) [rectangle,draw, fill=black!10] (p1) {$\theta$ students};
\node at (0,-1.2) [] (p2) {$\vdots$};
\node at (0,-1.8) [rectangle,draw, fill=black!10] (pm) {$\theta$ students};
\node at (0.25,-2.6) [rectangle,draw, fill=black!10] (pm1) {\enskip $m\theta$ students \quad~};
\node at (0.5,-3.4) [rectangle,draw, fill=black!10] (pm2) {\enskip $m\theta+m$ students \quad~};
\node[above = 0cm of p1] {Projects:};
\node[left = 0cm of p1] {$p_1$};
\node[left = 0cm of pm] {$p_m$};
\node[left = 0cm of pm1] {$p_{m+1}$};
\node[left = 0cm of pm2] {$p_{m+2}$};
\node at (5,-0.8) [rectangle,draw] (r1) {$\theta+1$};
\node at (5,-1.2) [] (r2) {$\vdots$};
\node at (5,-1.8) [rectangle,draw] (rm) {$\theta+1$};
\node at (4.85,-2.6) [rectangle,draw] (rm11) {$w_1$};
\node at (5.4,-2.6) [] (rm12) {$\ldots$};
\node at (6.1,-2.6) [rectangle,draw] (rm1m) {$w_{4m}$};
\node at (5.4,-3.4) [rectangle,draw] (rz) {$m\theta+m-1$};
\node[above = 0cm of r1] {Resources:};
\node[right = 0cm of r1] {$r_{x_1}$};
\node[right = 0cm of rm] {$r_{x_m}$};
\node[right = 0cm of rm1m] {$r_i$};
\node[right = 0cm of rz] {$r_z$};
\draw[->] (r1) -- (p1);
\draw[->] (rm) -- (pm);
\draw[->] (rm11) -- (pm1);
\draw[->] (rz) -- (pm2);
\draw[->,dashed,thin] (r1) to [out = 180, in = 0] (pm2);
\draw[->,dashed,thin] (rm) to [out = 180, in = 0] (pm2);
\draw[->,dashed,thick] (rm11) to [out = 180, in = 0] (p1);
\draw[->,dashed,thick] (rm11) to [out = 180, in = 0] (0.8,-1.3);
\draw[->,dashed,thick] (rm11) to [out = 180, in = 0] (pm);
\draw[->,dashed,thin] (rz) to [out = 180, in = 0] (pm1);
\node at (2.2,-1.3) [rectangle,draw,dashed,fill=black!05] {4-Partition?};
\end{tikzpicture}
\caption{Reducing \textsc{4-Partition} to \textsc{SPR/Nw/Verif}. Students specified in project boxes are the students that are acceptable for each project. While the horizontal resource allocation makes almost all capacity requirements feasible, one more student can be matched to $p_{m+2}$ if and only if the dashed resource allocation (with a solution to \textsc{4-Partition}) is feasible.}\label{fig:1}
\end{figure}

\begin{theorem}
\textsc{SPR/Nw/Verif} is coNP-complete,
even if each student only has one acceptable project.
\end{theorem}
\begin{proof}
Claiming pair $(s,p)$ and allocation $\mu'$ that makes it feasible are efficiently verifiable no-certificates. Hence, \textsc{SPR/Nw/Verif} is in coNP.
To show coNP-hardness, any instance $W=\{w_1,\ldots,w_{4m}\}$ of \textsc{4-Partition} with target $\theta$ (assuming $\sum_{w\in W} w = m\theta$ and $\frac{\theta}{5}<w_i<\frac{\theta}{3}$) is reduced to the following co-instance, whose yes-answers are for existent claiming pairs (see Fig. \ref{fig:1}). 
There are $m+2$ projects.
For $i\in[m]$, $\theta$ students only consider $p_i$ acceptable,
$m\theta$ students only consider $p_{m+1}$ acceptable,
and $m\theta+m$ only consider $p_{m+2}$ acceptable.
Projects also rank the corresponding students acceptable, arbitrarily.
In matching $Y$, all students are matched except one student $s^\ast$, who wanted $p_{m+2}$.
In allocation $\mu$, for every $i\in[m]$, project $p_i$ receives resource $r_{x_i}$ with capacity $q_{r_{x_i}}=\theta+1$ and $T_{r_{x_i}}=\{p_i,p_{m+2}\}$.
Project $p_{m+1}$ receives $4m$ resources $r_i$ identified with integer set $W=\{w_1,\ldots,w_{4m}\}$: for every $i\in[4m]$, resource $r_i$ has capacity $q_{r_{i}}=w_i$ and $T_{r_{i}}=\{p_i\mid i\in[m+1]\}$.
Project $p_{m+2}$ receives resource $r_z$ with capacity $q_{r_z}=m\theta+m-1$ and $T_{r_z}=\{p_{m+1},p_{m+2}\}$.
Since integers $w_i$ and $\theta$ are polynomially bounded, so is the number of students, and the reduction is polynomial-time.
There exists a solution $V_1,\ldots,V_m$ to \textsc{4-Partition} if and only if allocation  $\mu'$ (dashed in Figure \ref{fig:1}) is feasible, i.e. $(s^\ast,p_{m+2})$ is a claiming pair.
%
%
\end{proof}

\begin{theorem}
\textsc{SPR/Nw/Find} belongs to $\text{FP}^{\text{NP}}[\text{poly}]$ and is $\text{FP}^{\text{NP}}[\text{log}]$-hard,
even if each student only has a single acceptable project.
\end{theorem}

\begin{proof}
Mechanism SD shows that \textsc{SPR/Nw/Find} belongs to $\text{FP}^{\text{NP}}[\text{poly}]$. Hardness follows from Lemmas \ref{lem:1} and \ref{lem:2} below.
\end{proof}

\begin{lemma}\label{lem:1}
\textsc{ParetoPartition} is 
$\text{FP}^{\text{NP}}[\text{poly}]$-complete and 
strongly $\text{FP}^{\text{NP}}[\text{log}]$-hard.
\end{lemma}
\begin{proof}
\textsc{ParetoPartition} (a partition into $m$ subsets targeting $\theta$) belongs to $\text{FP}^{\text{NP}}[\text{poly}]$. Indeed a Leximax partition (thus Pareto efficient) can be found by making a polynomial number of calls to an NP-oracle on the following subproblem: Given one deficit per subset ${\delta}_1,\ldots,{\delta}_{m}$, decide whether a mapping from $W$ to subsets $V_1,\ldots,V_{m}$ exists, such that deficits are greater or equal to ${\delta}_1,\ldots,{\delta}_{m}$. A Leximax partition can be found by iterating on $V_i$ from $V_1$ to $V_{m}$. 
Assuming the first components ${\delta}_1,\ldots,{\delta}_{i-1}$ of a Leximax Pareto efficient partition were previously fixed by iterations $V_1$ to $V_{i-1}$ and ${\delta}_{i+1}=\ldots={\delta}_{m}=-\theta$,
we set ${\delta}_{i}$ to the best feasible deficit for $V_i$ by a binary search in $[-\theta,0]$ using the NP-oracle on the subproblem above.

Let any instance of \textsc{Max3DM} be defined by finite sets $A,B,C$ with $|A|=|B|=|C|=d$ and triplets set $N\subseteq A\times B\times C$, $|N|=n$. Triplet $t=(a,b,c)\in N$ is mapped to  payoff $v_t\in\mathbb{N}$. 
In a (partial) 3-dimensional matching (3DM) $N'\subseteq N$, any element of $A\cup B\cup C$ occurs at most once.
The goal is to maximize $\sum_{t\in N'}v_t$ for $N'\subseteq N$ any (partial) 3-dimensional matching. Note that maximizing $-\sum_{t\in N'\setminus N}v_t$ is an equivalent goal.
This problem is $\text{FP}^{\text{NP}}[\text{poly}]$-complete \cite[Th. 3.5]{gasarch1995optp}.
For every $a\in A$ (resp. $b\in B$, $c\in C$), let $\#a$ (resp. $\#b$, $\#c$) denote the number of occurrences of $a$  (resp. $b$, $c$) in $N$: the number of triplets that contain $a$ (resp. $b$, $c$). Let $v_N$ denote total $\sum_{t\in N}v_t$. Elements $a_i\in A,b_j\in B,c_k\in C$ and triplets $t\in N$ are identified with integers $i,j,k\in[d]$ and $t\in[n]$.

We reduce this problem to the following instance of \textsc{ParetoPartition} for which finding a Pareto efficient solution produces the optimum for the given \textsc{Max3DM} instance. 
Set $W$ contains $8n$ integers that must be partitioned into $m=n+1$ subsets (of various cardinalities). 
Given basis $\beta\in\mathbb{N}_{\geq 2}$ and integer sequence $(z_i)_{i\in\mathbb{N}}$, we define integer $\langle \ldots z_2~z_1~z_0\rangle$ by $\sum_{i\geq 0}z_i\beta^i$.
Let $\beta$ be an integer large enough for such representation in basis $\beta$ to never have carryovers, even when one adds all the integers in $W$. 
Choosing $\beta=\max\{30n^3d,nv_N\}+1$ largely satisfies this purpose. 
Let $\Sigma_n$ denote $\sum_{t=1}^{n}t=\frac{n(n+1)}{2}$.
The integers in set $W$ are represented below. 
For each $t=(a_i,b_j,c_k)\in N$, there is a \emph{triplet}-integer $w(t)$.
For each $a_i\in A$, we introduce one \emph{actual}-integer $w(a_i)$ representing the actual element intended to go with the triplets in a (partial) 3DM, and $\#a_i-1$ \emph{dummies}, present in triplets that are not in the 3DM. Similarly, we introduce $\#b_j$ (resp. $\#c_k$) integers for every $b_j\in B$ (resp. $c_k\in C$).
For each $t\in N$, there are four \emph{value}-integers $w(v_t)$. 
Target $\theta$ is below. We also indicate values $\theta-w(t)$ which will be useful later. 
$$
\begin{array}{lrcccccccccl}

 & & z_7 & z_6 & z_5 & z_4 & z_3 & z_2 & z_1 & z_0\\[1ex]
\hline

\forall t\!\in\! N, & 
w(t\!=\!(a_ib_jc_k)) = \langle & {3n\!-\!4} & {24n\!-\!15} & -i & -j & -k & {3\Sigma_n\!-\!t} & 3d\!+\!3n\!-\!3 & v_N&\rangle\\[2ex]

\forall a_i\!\in\! A,&
^{\text{one actual,}}_{\#a_i-1\text{ dum.}} ~ w(a_i) = \langle& 1 & 1 & i & 0 & 0 & 0 & ^{1\text{ (actual)}}_{0\text{ (dum.)}} & 0 &\rangle\\[2ex]

\forall b_j\!\in\! B,&
^{\text{one actual,}}_{\#b_j-1\text{ dum.}} ~ w(b_j) = \langle& 1 & 2 & 0 & j & 0 & 0 & ^{1\text{ (actual)}}_{0\text{ (dum.)}} & 0 &\rangle\\[2ex]

\forall c_k\!\in\! C,&
^{\text{one actual,}}_{\#c_k-1\text{ dum.}} ~ w(c_k) = \langle& 1 & 4 & 0 & 0 & k & 0 & ^{1\text{ (actual)}}_{0\text{ (dum.)}} & 0 &\rangle\\[2ex]

\forall t\!\in\! N,&
 ^{^{\text{``zero''}}_{\text{``one''}}}_{^{\text{``two''}}_{\text{``three''}}} ~ w(v_t) = \bigg\langle& 1 & 8 & 0 & 0 & 0 & t & ^{^{0\text{ (zero)}}_{1\text{ (one)}}}_{^{2\text{ (two)}}_{3\text{ (three)}}} & ^{^{-v_t}_{0}}_{^{0}_{0}} &\bigg\rangle\\[2ex]
\hline
\text{Target } & \theta = \langle& 3n & 24n & 0 & 0 & 0 & 3\Sigma_n & 3d+3n & 0 & \rangle\\[2ex]
\hline
\textit{Remark:} & \theta-w(t) = \langle& 4 & 15 & i & j & k & t & 3 & -v_N & \rangle
\end{array}
$$

Since every subset has same target $\theta$, 
given a partition $\left(V_i\mid i\in[m]\right)$ with deficits $\bm{\delta}\in\mathbb{Z}^m$ and any permutation $\sigma:[m]\leftrightarrow[m]$, deficits $(\delta_{\sigma(i)}\mid i\in[m])$ are also feasible by the permuted partition $(V_{\sigma(i)}\mid i\in[m])$. On every column but $z_0$, total offer (weights) equates total demand (targets). For instance, on column $z_1$, it holds that $n(3d+3n-3)+3d+6n=(n+1)(3d+3n)$.

Given any maxim\underline{al} 3DM $N'\subseteq N$, one can make partitions such that for one arbitrary subset $V_{(\ast)}$ deficit is $\delta_{(\ast)}=-\sum_{t\in N\setminus N'}v_t$  and for the $n$ other subsets $V_{(t)}$ deficit is $\delta_{(t)}=0$, as follows:
\begin{itemize}
\item For every $t=(a_i,b_j,c_k)\in N'$, we make a subset $V_{(t)}$ that contains $w(t)$, the three actuals $w(a_i),w(b_j),w(c_k)$  and integer $w(v_t)$ ``zero''. 
Integers $w(v_t)$ ``one, two and three'' are sent to $V_{(\ast)}$ without the $-v_t$ deficit from ``zero''.
\item For every $t=(a_i,b_j,c_k)\in N\setminus N'$, we make a subset $V_{(t)}$ that contains $w(t)$, actual or dummy integers $w(a_i),w(b_j),w(c_k)$;\footnote{A complete 3DM may not exist. A partial 3DM may leave some actuals in $N\setminus N'$.}
and, if subset $V_{(t)}$ contains resp. one, two or three dummies,\footnote{By maxim\underline{al}ity of $N'$, zero dummies is not possible.} integer $w(v_t)$ respectively ``one, two or three''. The other integers $w(v_t)$ which include deficit $-v_t$ are sent to $V_{(\ast)}$.
\end{itemize}
From any optimal 3DM $N'$ and $i\in[m]$, let $\bm{\delta}^{\text{opt}(i)}$ be the deficit vector ${\delta}^{\text{opt}(i)}_{i}=-\sum_{t\in N\setminus N'}v_t$ and  $\bm{\delta}^{\text{opt}(i)}_{-i}\equiv 0$\footnote{Given a vector $\bm{\delta}\in\mathbb{Z}^{n+1}$ and $i\in[n+1]$, $\bm{\delta}_{-i}\in\mathbb{Z}^{n}$ denotes the same vector where the $i$th component is removed.} where $V_{(\ast)}=V_i$.
Below, we show that this family of $m$ deficit vectors dominate all the others, hence are the only Pareto efficient ones.
The idea is that every subset $V_i$ (which objective is to maximize $\delta_i$ up to zero), has a column-wise lexicographic preference on integers, from heaviest column $z_7$ (weight $\beta^7$) to the lowest $z^0$ (weight $\beta^0$). Indeed, since in each column (but $z_0$), total offer (weights) equates total demand (targets), an unbalanced partition is always dominated: at efficiency, a column's deficit is exactly zero and cannot overrun a lower one.  And, sums of integers in $W$ never have carryovers from a column to a heavier one.
By reasoning from $z^7$ to $z^1$, any partition which does not satisfy all the following conditions is clearly Pareto dominated by some $\bm{\delta}^{\text{opt}(i)}$ because of one huge deficit in multiples of $\beta$ on some component $\delta_i$.
\begin{description}
\setlength{\itemsep}{1em}
\item[$z^7$: ] No subset contains two triplet-integers. 
Therefore, $n$ subsets (among $m\!=\!n\!+\!1$) can be identified from the triplet-integer $w(t)$ contained by $V_{(t)}$; and we identify the last one by $V_{(\ast)}$. These subsets can be ordered indifferently. For a subset $V_{(t)}$, remaining deficit $\theta-w(t)$ is:
$$
\arraycolsep=6.0pt
\begin{array}{rcccccccc}
\theta_{(t)}=\langle 4 & 15 & i & j & k & t & 3 & -v_N \rangle
\end{array}
$$
Thus, subsets $V_{(t)}$ must contain four other integers to cancel the deficit 4
at $z^7$. Then, $V_{(\ast)}$ must contain the remainder of the $n$ integers;
the deficit on column $z^7$ becomes $0$. 
\item[$z^6$:]
Subset $V_{(\ast)}$ contains $3n$ \emph{value}-integers $w(v_t)$,
so its value at $z^6$ be $24n$ (i.e., no deficit).  
Subsets $V_{(t)}$ must contain one of each in integers $w(a)$, $w(b)$, $w(c)$ and $w(v)$
to cancel deficit 15 at $z^6$.
\item[$z^{5}$--$z^{2}$:]
To cancel deficits from $z^5$ to $z^2$,
for $t\!=\!(a_i,b_j,c_k)\!\in\!N$, subset $V_{(t)}$ contains
\emph{precisely} one of each in integers $w(a_i)$, $w(b_j)$, $w(c_k)$,
and $w(v_t)$.
Also, $V_{(\ast)}$ has deficit $3\sum_{n}$ at $\beta^2$.
Thus, it needs \emph{exactly} three $w(v_t)$ of every triplet $t$ to
cancel the deficit, otherwise some $V_{(t)}$ would be missing his.
\item[$z^{1}$--$z^{0}$:]
Again, due to tightness of offer on demand for $z_1$,
subset $V_{(t)}$ must contain either (i) three actual elements and integer $w(v_t)$ ``zero'' or
(ii) one, two or three dummy elements and integer $w(v_t)$ respectively ``one, two or three''.
In case (i), integers $w(v_t)$ one, two and three go to $V_{(\ast)}$ without degrading it.
In case (ii), three integers $w(v_t)$ which include integer $w(v_t)$ ``zero'' go to $V_{(\ast)}$ and degrade it by $-v_t$.
\end{description}

All in all, Pareto efficiency constrains partitions to structure as in the mapping from a 3-dimensional matching $N'$ given above: the only possible Pareto efficient deficit vectors are $\bm{\delta}^{\text{opt}(i)}$ for $i\in[m]$ and thus provide the optimum for \textsc{Max3DM}. Consequently, this reduction is metric.
Since weighted \textsc{Max3DM} is $\text{FP}^{\text{NP}}[\text{poly}]$-hard \cite{gasarch1995optp}, so is \textsc{ParetoPartition}.
Since 
\emph{unweighted} \textsc{Max3DM} is $\text{FP}^{\text{NP}}[\text{log}]$-hard 
and for $v_t\!\in\!\{0,1\}$ no integer exceeds polynomial $\beta^8$, 
\textsc{ParetoPartition} is also \emph{strongly} $\text{FP}^{\text{NP}}[\text{log}]$-hard. 
\end{proof}

\begin{lemma}\label{lem:2}
If the numbers in problem \textsc{ParetoPartition} are polynomially bounded, then the reduction \textsc{ParetoPartition} $\leq_p$ \textsc{SPR/Nw/Find} holds.
\end{lemma}
\begin{proof}
We reduce any instance  $W=\{w_1,\ldots,w_{|W|}\}$, $m\in\mathbb{N}$, $\theta\in\mathbb{N}$ of \textsc{ParetoPartition} to an instance of \textsc{SPR/Nw/Find}. There are $m$ projects $p_1,\ldots,p_m$; for every project $p_i$ there is a disjoint set of $\theta$ students who consider only $p_i$ acceptable (and reciprocally). Project $p_i$ ranks these students arbitrarily.  Resources $R$ are identified with set $W$: any resource is compatible with any project and $q_R=(w_1,\ldots,w_{|W|})$. Crucially, with numbers in  \textsc{ParetoPartition} polynomially bounded, there are only polynomially many students.

Computing a nonwasteful matching $(Y,\mu)$ outputs a partition $V_1,\ldots,V_m\equiv \mu^{-1}(p_1),\ldots,\mu^{-1}(p_m)$ with Pareto efficient deficits.
Indeed, by definition, a claiming pair would exist if and only if there was an allocation (resp. partition) where the number of unmatched students per project (resp. deficit vector) Pareto dominated the ``deficit vector'' of allocation/partition $V_1,\ldots,V_m$.
\end{proof}

\subsection{The Complexity of Stability}

A matching that is both nonwasteful and fair (i.e., stable) may not exist. In this section, we settle the complexity of deciding whether such a matching exists in a given SPR as $\text{NP}^{\text{NP}}$-complete, which is strictly more intractable than NP-complete.

\begin{theorem}\label{th:stable:verif}
\textsc{SPR/Stable/Verif} is coNP-complete,
even if students only have one acceptable project.
\end{theorem}
\begin{proof}
The construct is the same as for \textsc{SPR/Nw/Verif}.
Assuming that in the given matching project $p_{m+2}$ has its $m\theta+m-1$ top-preferred students,
the concept of an envious pair becomes empty in this construction; hence stability amounts to nonwastefulness. Therefore, the same proof holds. 
\end{proof}

\begin{theorem}
\textsc{SPR/Stable/Exist} is $\text{NP}^{\text{NP}}$-complete.
\end{theorem}

\begin{proof}
A stable matching is a yes-certificate verifiable by NP-oracle (Theorem \ref{th:stable:verif});
hence, \textsc{SPR/Stable/Exist} belongs to $\text{NP}^{\text{NP}}$.
Hardness follows from Lemmas \ref{lem:3} and \ref{lem:4} below.
\end{proof}

\begin{lemma}\label{lem:3}
\textsc{$\forall\exists$-4-Partition} is strongly $\text{coNP}^{\text{NP}}$-hard.
\end{lemma}
\begin{proof}
Let any instance of \textsc{$\forall\exists$-3DM} be defined by finite sets $A,B,C$ with $|A|=|B|=|C|=d$ and two disjoint triplet sets $M,N\subseteq A\times B\times C$, with $|M|=n'$ and $|N|=n$. 
This decision problem asks the following question:
$$
\forall M'\subseteq M,\quad
\exists N'\subseteq N,\quad
M'\cup N'\text{ is a 3DM,}
$$
where ``$M'\cup N'$ is a 3DM'' means that any element of $A\cup B\cup C$ occurs exactly once in $M'\cup N'$.
  This is a $\text{coNP}^{\text{NP}}$-complete problem \citep{mcloughlin1984}.
For every $a_i\in A$ (resp. $b_j\in B$, $c_k\in C$), let $\#a_i$ (resp. $\#b_j$, $\#c_k$) denote the number of occurrences of $a_i$  (resp. $b_j$, $c_k$) in $M\cup N$: how many triplets contain $a_i$ (resp. $b_j$, $c_k$)? We identify elements and triplets with integers $i,j,k\in[d]$ and $t\in[n'+n]$.

We reduce this instance to the following \textsc{$\forall\exists$-4-Partition} instance. List $W$ contains the $4(n'+n)$ integers depicted below in basis $\beta=4(n'+n)d+1$ (definition in proof of Lemma \ref{lem:1}). For every triplet $t=(a_i,b_j,c_k)\in M\cup N$, there is one ``triplet'' integer $w(a_i,b_j,c_k)\in\mathbb{N}$.
For every element $a\in A$, we introduce one \emph{actual} integer $w(a)$ that represents the actual element intended to go with the triplets in the 3DM and $\#a-1$ \emph{dummies} that will go with the triplets that are not in the 3-dimensional matching.
Similarly, we introduce $\#b$ integers for each $b\in B$ and $\#c$ integers for each $c\in C$. Target $\theta=4\beta^5+15\beta^4$ is also depicted below.
Numbers are polynomially bounded by $\beta^6$.
$$
\arraycolsep=10pt
\begin{array}{lrcccccccl}

\forall t\in M,&
w({t=(a_ib_jc_k)}) = \langle & 1 & 1 & -i & -j & -k & 0& \rangle\\[1ex]

\forall a_i\in A,&
^{\quad\text{one actual}}_{\#a_i\!-\!1\text{ dum.}} ~w(a_i) = \langle& 1 & 2 & i & 0 & 0 & ^{-2\text{ (actual)}}_{~~0\text{ (dummy)}}& \rangle\\[1ex]

\forall b_j\in B,&
^{\quad\text{one actual}}_{\#b_j\!-\!1\text{ dum.}} ~w(b_j) = \langle& 1 & 4 & 0 & j & 0 & ^{+1\text{ (actual)}}_{~~0\text{ (dummy)}}& \rangle\\[1ex]

\forall c_k\in C,&
^{\quad\text{one actual}}_{\#c_k\!-\!1\text{ dum.}} ~w(c_k) = \langle& 1 & 8 & 0 & 0 & k & ^{+1\text{ (actual)}}_{~~0\text{ (dummy)}}& \rangle\\[1ex]
\hline
\textbf{target } & \theta = \langle& 4 & 15 & 0 & 0 & 0 & 0 & \rangle
\end{array}
$$
List $\mathcal{L}$ has length $\ell=|M|$: 
every triplet $t=(a_i,b_j,c_k)\in M$ is reduced to couple $u_tv_t$ between ``triplet'' integer $u_t=w(a_i,b_j,c_k)$ and ``actual'' integer $v_t=w(a_i)$. 

First, since $\beta$ is large enough, 
and column-wise offer (weights) equates demand (targets),
additions in
$W$ never have carryovers. 
Therefore, subsets must hit the target on each
column of this representation. Consequently, in any 4-partition of
$W$, there are four elements, one of each in the following: ``triplet'' integers, element-$a$ integers, element-$b$ integers and element-$c$ integers. Moreover, ``triplet'' integer $w(a_i,b_j,c_k)$ is with ``its'' elements $w(a_i)$, $w(b_j)$ and $w(c_k)$. Also, \emph{actual} elements must be in the same subset and dummies in the others. 
Therefore, any 3-dimensional matching $M'\cup N'$ is in correspondence with such a 4-partition.
Validity follows from the correspondence between $M'$ (taking or not elements in $M$) and $\sigma$ (enforcing integers $w(t)$ for $t\in M$ in the same subsets as its actual elements $w(a_i)$ and the two others.)

(yes$\Rightarrow$yes) Assume the 3DM instance is a yes one, and let $\sigma:[\ell]\rightarrow\{0,1\}$ be any couple enforcement/forbidding function. We construct a $\sigma$-satisfying 4-partition in correspondence with the following 3-dimensional matching $M'\cup N'$: for $t\in[\ell]\!\equiv\!M$, triplet $t$ is in $M'$ if and only if $\sigma(t)=1$; then the assumption gives $N'$ such that  $M'\cup N'$ is a 3DM. We construct the corresponding 4-partition (see paragraph above), and it is $\sigma$-satisfying. 

(yes$\Leftarrow$yes) Assume the partition instance is a yes one, and let us show that $\forall M'\subseteq M, \exists N'\subseteq N$ s.t. $M'\cup N'$ is a 3DM. Given $M'$, let $\sigma$ be defined as $\sigma(t)=1$ if and only if $t\in M'$. A $\sigma$-satisfying 4-partition exists, and is in correspondence with some 3DM $M'\cup N'$, by construction, as above.
\end{proof}

\begin{lemma}\label{lem:4}
\textsc{$\forall\exists$-4-Partition} $\leq_p$ \textsc{co-SPR/Stable/Exist}
\end{lemma}

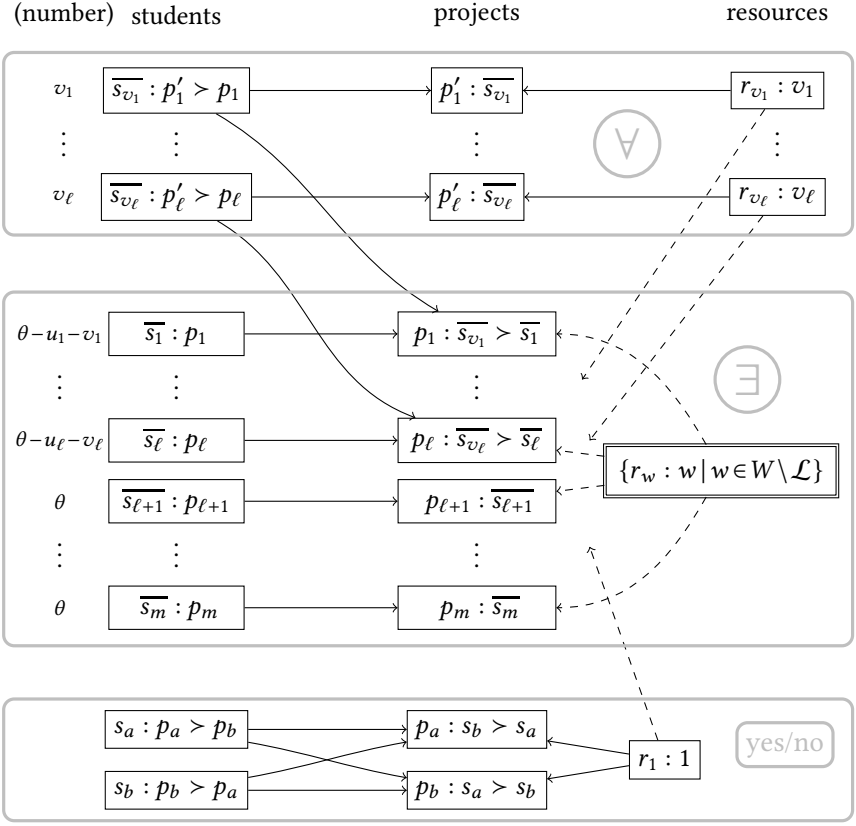
\begin{figure}[t]
\centering
\begin{tikzpicture}
\node at (0,1.0) [rectangle,draw] (p1p) {$p'_1:\overline{s_{v_1}}$};
\node at (0,0.4) [] (pdp) {$\vdots$};
\node at (0,-0.4) [rectangle,draw] (plp) {$p'_\ell:\overline{s_{v_\ell}}$};
\node at (4,1.0) [rectangle,draw] (rv1) {$r_{v_1}:v_1$};
\node at (4,0.4) [] (rdv) {$\vdots$};
\node at (4,-0.4) [rectangle,draw] (rvl) {$r_{v_\ell}:v_\ell$};
\node at (-4,1.0) [rectangle,draw] (sv1) {$\overline{s_{v_1}}:p'_1\succ p_1$};
\node at (-5.5,1.0)[]{\footnotesize $v_1$};
\node at (-4,0.4) [] (sd1) {$\vdots$};
\node at (-5.5,0.4) [] (sd1) {$\vdots$};
\node at (-4,-0.4) [rectangle,draw] (svl) {$\overline{s_{v_\ell}}:p'_\ell\succ p_\ell$};
\node at (-5.5,-0.4)[]{\footnotesize $v_\ell$};
\node at (0,-2.2) [rectangle,draw,minimum width=2.1cm] (p1) {$p_1:\overline{s_{v_1}} \succ \overline{s_1}$};
\node at (0,-2.8) [] (pd1) {$\vdots$};
\node at (0,-3.6) [rectangle,draw,minimum width=2.1cm] (pl) {$p_\ell:\overline{s_{v_\ell}} \succ \overline{s_\ell}$};
\node at (-4,-2.2) [rectangle,draw,minimum width=1.8cm] (s1) {$\overline{s_1}:p_1$};
\node at (-5.55,-2.2)[]{\footnotesize $\theta\!-\!u_1\!-\!v_1$};
\node at (-4,-2.8) [] (sd1) {$\vdots$};
\node at (-5.55,-2.8) [] (sd1) {$\vdots$};
\node at (-4,-3.6) [rectangle,draw,minimum width=1.8cm] (sl) {$\overline{s_\ell}:p_\ell$};
\node at (-5.55,-3.6)[]{\footnotesize $\theta\!-\!u_\ell\!-\!v_\ell$};
\node at (0,-4.4) [rectangle,draw,minimum width=2.1cm] (pl1) {~~$p_{\ell+1}:\overline{s_{\ell+1}}$~~};
\node at (0,-5.0) [] (pd2) {$\vdots$};
\node at (0,-5.8) [rectangle,draw,minimum width=2.1cm] (pm) {~~~~$p_m:\overline{s_m}$~~~~};
\node at (-4,-4.4) [rectangle,draw,minimum width=1.8cm] (sl1) {$\overline{s_{\ell+1}}:p_{\ell+1}$};
\node at (-5.55,-4.4)[]{\footnotesize $\theta$};
\node at (-4,-5.0) [] (sd1) {$\vdots$};
\node at (-5.55,-5.0) [] (sd1) {$\vdots$};
\node at (-4,-5.8) [rectangle,draw,minimum width=1.8cm] (sm) {~~~$\overline{s_m}:p_m$~~~};
\node at (-5.55,-5.8)[]{\footnotesize $\theta$};
\node at (3.25,-4) [rectangle,double, draw,inner sep=0.5em] (rw) {$\{r_w:w\!\mid\!w\!\in\!W\!\setminus\!\mathcal{L}\}$};
\node at (0,-7.4) [rectangle,draw] (pa) {$p_a:s_b\succ s_a$};
\node at (0,-8.2) [rectangle,draw] (pb) {$p_b:s_a\succ s_b$};
\node at (2.5,-7.8) [rectangle,draw] (rd) {$r_1:1$};
\node at (-4,-7.4) [rectangle,draw] (sa) {$s_a:p_a\succ p_b$};
\node at (-4,-8.2) [rectangle,draw] (sb) {$s_b:p_b\succ p_a$};
\path[->]
	(sv1) edge[out=0,in=180] (p1p)
	(sv1) edge[out=330,in=150] (p1)
	(svl) edge[out=0,in=180] (plp)
	(svl) edge[out=330,in=160] (pl)
	(s1) edge[out=0,in=180] (p1)
	(sl) edge[out=0,in=180] (pl)
	(sl1) edge[out=0,in=180] (pl1)
	(sm) edge[out=0,in=180] (pm)
	(sa) edge[out=0,in=180] (pa)
	(sa) edge[out=350,in=170] (pb)
	(sb) edge[out=10,in=190] (pa)
	(sb) edge[out=0,in=180] (pb);
\path[->]
	(rv1) edge[out=180,in=0] (p1p)
	(rv1) edge[dashed]  (1.4,-2.8) 
	(rvl) edge[out=180,in=0] (plp)
	(rvl) edge[dashed]  (1.5,-3.6)
	(rw) edge[out=120,in=0,dashed] (p1)
	(rw) edge[dashed] (pl)
	(rw) edge[dashed] (pl1)
	(rw) edge[out=240,in=0,dashed] (pm)
	(rd) edge[dashed] (1.5,-5.0) 
	(rd) edge (pa)
	(rd) edge (pb);
\draw[draw=black!25,very thick, rounded corners] (-6.3,1.5) rectangle (5,-0.9);
\draw[draw=black!25,very thick, rounded corners] (-6.3,-1.65) rectangle (5,-6.3);
\draw[draw=black!25,very thick, rounded corners] (-6.3,-7) rectangle (5,-8.6);
\node at (2.0,0.3) [circle,draw,black!25,very thick] {\huge $\forall$};
\node at (3.6,-2.8) [circle,draw,black!25,very thick] {\huge $\exists$};
\node at (4.1,-7.6) [rectangle, rounded corners,draw,black!25,very thick] {\large yes/no};
\node at (-5.5,2) {(number)};
\node at (-4,2) {students};
\node at (0,2) {projects};
\node at (4,2) {resources};
\end{tikzpicture}
\caption{From \textsc{$\forall\exists$-4-Partition} to \textsc{co-SPR/Stable/Exist}. Left-right arrows depict acceptable projects and right-left arrows, compatible projects. Dashed arrows go to any project $p_1\ldots p_m$, but $p_j$ for resource $r_{v_j}$.}\label{fig:2}
\end{figure}

\begin{proof}
Given a \textsc{$\forall\exists$-4-Partition} instance defined by $m\in\mathbb{N}$, list $W=\{w_1,\ldots,w_{4m}\}$, target $\theta\in\mathbb{N}$, and list of couples $\mathcal{L}=(u_1,v_1),\ldots,(u_\ell,v_\ell)$ of $W$, we construct a \textsc{co-SPR/Stable/Exist} instance depicted in Figure \ref{fig:2}. It contains:
\begin{itemize}
\item $\ell+m+2$ projects $p'_1,p'_2\ldots,p'_\ell$,\enskip $p_1,p_2,\ldots,p_m$\enskip and $p_a,p_b$,
\item $\ell$ subsets of students $\overline{s_{v_1}},\overline{s_{v_2}},\ldots,\overline{s_{v_\ell}}$ where each subset $\overline{s_{v_i}}$ contains $v_i$ students who all have preference $\overline{s_{v_i}}:p'_i\succ p_i\succ \emptyset$,
\item $m$ subsets of students $\overline{s_1},\overline{s_2},\ldots,\overline{s_\ell},\overline{s_{\ell+1}},\ldots,\overline{s_m}$ where each subset $\overline{s_{i}}$ for $i\in[\ell]$ contains $\theta-u_i-v_i$ students, each subset $\overline{s_{i}}$ for $i\in[\ell+1,m]$ contains $\theta$ students and in every subset $\overline{s_{i}}$ students all have preference $\overline{s_{i}}:p_i\succ\emptyset$, and
\item two students $s_a,s_b$ who have preferences $s_a:p_a\succ p_b\succ\emptyset$ and $s_b:p_b\succ p_a\succ\emptyset$.
\item For every $i\in[\ell]$, project $p'_i$ has preference $p'_i:\overline{s_{v_i}}\succ\emptyset$, and project $p_i$ has preference $p'_i:\overline{s_{v_i}}\succ\overline{s_{i}}\succ\emptyset$. 
For every $i\in[\ell+1,m]$, project $p_i$ has preference $p_i:\overline{s_{i}}\succ\emptyset$.
Project $p_a$ has preference $p_a:s_b\succ s_a$ and $p_b$ preference $p_b:s_a\succ s_b$  (as in Example \ref{counterex:stable}).
\end{itemize}
Since \textsc{$\forall\exists$-4-Partition} is strongly hard, we can assume that its numbers are polynomially bounded (e.g. w.r.t. $m$); hence, there is a polynomial number of students. 
There are $|W|-\ell+1$ resources:
\begin{itemize}
\item for every $i\in[\ell]$, resource $r_{v_i}$ has capacity $q_{r_{v_i}}=v_i$ and is compatible with $\{p'_i\}\cup\{p_j\mid j\neq i\}$,
\item for every weight $w\in W\setminus\mathcal{L}$ resource $r_w$ has capacity $q_{r_w}=w$ and $T_{r_w}=\{p_1,\ldots,p_m\}$, and
\item resource $r_1$ has capacity $q_{r_1}=1$ and compatibilities $T_{r_1}=\{p_a,p_b\}\cup\{p_1,\ldots,p_m\}$.
\end{itemize}

The idea is that capacity requirements of projects $p_1,\ldots,p_m$ model the $m$ targets of a 4-partition. Since integers $u_1,\ldots,u_\ell$ are in $V_1,\ldots,V_\ell$, we already subtract them from $p_1,\ldots,p_\ell$. The universal quantifier is encoded as follows.
\begin{itemize}
\setlength{\itemindent}{5mm}
\item[\underline{$\sigma(i)\!=\!1$:}] Enforcing $u_i$ and $v_i$ together in a 4-partition will correspond to letting the capacity requirement of project $p_i$ be $\theta-u_i-v_i$ (like if $u_i$ and $v_i$ were already inside): students $\overline{s_{v_i}}$ are matched with $p'_i$ and resources $r_{v_i}$ are allocated to $p'_i$. 
\item[\underline{$\sigma(i)\!=\!0$:}] Conversely, forbidding $u_i$ and $v_i$ to be together in a 4-partition will correspond to trying to match $\overline{s_{v_i}}$ with $p_i$, hence bringing its capacity requirement to $\theta-u_i$, while resource $r_{v_i}$ cannot be allocated to $p_i$. 
\end{itemize}
We are now set to formally prove the validity of this reduction.

(yes$\Rightarrow$yes)
For each $\sigma:[\ell]\rightarrow\{0,1\}$, there is a $\sigma$-satisfying 4-partition $V_1,\ldots,V_{\ell},V_{\ell+1},\ldots,V_m$.
For the sake of contradiction, let us assume that there exists a stable matching $(Y,\mu)$.
By definition, for each resource $r_{v_i}, i\in[\ell]$, either (1) $\mu(r_{v_i})=p'_i$ or (2) $\mu(r_{v_i})\in\{p_j\mid j\neq i\}$.

Let us consider a particular mapping $\sigma$ defined by $\sigma(i)=1$ if (1), and $\sigma(i)=0$ if (2).
By premise, there exists a $\sigma$-satisfying 4-partition $V_1,\ldots,V_{\ell},V_{\ell+1},\ldots,V_m$: for every $i\in[\ell]$, first $u_i\in V_i$ and second $v_i\in V_i$ if and only if $\sigma(i)=1$.  
From this $\sigma$-satisfying 4-partition, there exists an allocation of $\{r_{v_i}\mid \sigma(i)=0\}$ and $\{r_w\mid w\in W\setminus\mathcal{L}\}$ to projects $p_1,\ldots,p_m$ that makes feasible the full matching $Y(\overline{s_i})=p_i,\forall i\in[m]$ and $Y(\overline{s_{v_i}})=p_i,\forall i\in[m]\text{ s.t. }\sigma(i)=0$.
Therefore it would be wasteful to use resource $r_1$ on projects $\{p_1,\ldots,p_m\}$, contradicting stability, and consequently $r_1$ is allocated to $p_a$ or $p_b$. The SPR defined by $s_a,s_b,p_a,p_b,r_1$ cannot be stable (as in Example \ref{counterex:stable}). Consequently, a stable matching is impossible.

(no$\Rightarrow$no)
Assume that there exists a mapping $\sigma$ such that no $\sigma$-satisfying 4-partition exists, 
and let us build a stable matching $(Y,\mu)$ as follows. For every $i\in[\ell]$:
\begin{itemize}
\item if $\sigma(i)=1$, then $Y(\overline{s_{v_i}})=\{p'_i\}$ and $\mu(r_{v_i})=p'_i$;
\item if $\sigma(i)=0$, then $Y(\overline{s_{v_i}})=\{p_i\}$ and $\mu(r_{v_i})\!\in\!\{p_j\mid j\!\neq\!i\}$.
\end{itemize}
Then, we allocate the other resources ($\{r_w\mid w\in W\setminus\mathcal{L}\}$ \emph{and} $r_1$) in a way that minimizes the number of unmatched students in $\overline{s_1},\ldots,\overline{s_m}$.

The students from $\overline{s_{v_1}},\ldots,\overline{s_{v_\ell}}$ cannot be involved in a claiming (or envious) pair since they obtain their top choice if matched to $p'_i$, and one claiming pair from $p_i$ to $p'_i$ would deprive $p_1,\ldots,p_m$ from resource $r_{v_i}$ (not allocated to $p_i$), which is not feasible.
Since the number of unmatched students in $\overline{s_1},\ldots,\overline{s_m}$ is minimized, one more seat is not possible.
Since no $\sigma$-satisfying 4-partition exists, without resource $r_1$, some projects in $p_1,\ldots,p_m$ would loose a seat. Then $r_1$ cannot be re-allocated to $p_a$ or $p_b$ without canceling a seat. The remaining SPR defined by $s_a,s_b,p_a,p_b$ has no resource at all and is therefore stable.
\end{proof}

\section{Related Work}

This paper follows a stream of works
  dealing with constrained matching. 
Two-sided matching has been 
attracting considerable attention from AI and TCS
researchers~\citep{aziz2017stable,hamada2017weighted,hosseini2015manipulablity,kawase2017near}.
A standard market deals with maximum quotas, i.e., capacity limits that
cannot be exceeded.
However, many real-world matching markets are subject to a variety of
distributional constraints~\citep{kty:2014}, including regional maximum quotas, which
restrict the total number of students assigned to a set of schools%
~\citep{kamakoji-basic}, minimum quotas, which guarantee that 
a certain number of students are assigned to each school%
~\citep{fragiadakis::2012,Goto:aamas:2014,kurata:aaams2016,%
sonmez_switzer2013,sonmez_rotc2011}, and 
diversity constraints%
~\citep{hafalir2013effective,ehlers::2012,kojima2012school,kurata:jair2017}. 
  Other works examine the
computational complexity for finding a matching with desirable properties under distributional constraints, 
including~\citep{%
biro:tcs:2010,Fleiner16,hamda:esa:2011}.
A similar model was recently considered \citet{ismaili2018prima}, but with a compact representation scheme which handles exponentially many students and induces intrinsically different computational problems.

  There exist several works on three-sided matching problems~\cite{%
alkan1988,NgH91,huang2007} where three types of players/agents, e.g., males, females, and
pets, are matched. 
Although their model might look superficially similar to our model,
they are fundamentally different.
In the student-project allocation problem \citet{abraham2007two}, students are matched to projects,
while each project is offered by a lecturer. A student has a preference
over projects, and a lecturer has a preference over students. 
Each lecturer has her capacity limit. 
This problem can be considered as a standard two-sided matching
problem with distributional constraints. More specifically, 
this problem is equivalent to
a two-sided matching problem with regional maximum quotas \citet{kty:2014}.
A 3/2-approximation algorithm exists for the student-project allocation problem \cite{CooperM18}, and one can also obtain super-stability, despite ties \cite{OlaosebikanM18}.
 In our model, 
a resource is not an agent/player; it has no preference over
projects/students.
Also, a project/student has no preference over resources;
a project just needs to be allocated enough resources to accommodate
applying students.



\bibliographystyle{ACM-Reference-Format}
\bibliography{MAIN}

\end{document}